\documentclass[11pt]{article}

\usepackage{cite}
\usepackage{booktabs}
\usepackage{hyperref}
\usepackage{graphicx}
\usepackage{subfig}
\usepackage{paralist}

\usepackage{geometry}
\usepackage{amsthm}
\usepackage{amsmath}
\usepackage{amssymb}

\allowdisplaybreaks[1]

\newtheorem{theorem}{Theorem}
\newtheorem{lemma}[theorem]{Lemma}
\newtheorem{corollary}[theorem]{Corollary}
\newtheorem{observation}[theorem]{Observation}

\makeatletter
\def\@endtheorem{\endtrivlist}
\makeatother

\newcommand{\D}[2]{\mathcal{D}_{#1,#2}} 
\newcommand{\TL}[2]{#1^{#2}} 
\newcommand{\V}[1]{V(#1)} 
\newcommand{\Va}[1]{V_\alpha(#1)} 
\newcommand{\wdist}{w_1}
\newcommand{\wdistroot}{w_2}
\newcommand{\wdistspecial}{w_3}
\newcommand{\idxmin}{1}
\newcommand{\idxmax}{3}

\newcommand{\wnearestd}[1]{\mathrm{wnearest_d}(#1)}
\newcommand{\wnearestnd}[1]{\mathrm{wnearest_{nd}}(#1)}

\newcommand{\select}[3]{\mathrm{select}_{#1}(#2,#3)}
\newcommand{\tselect}{t_{\mathrm{select}}}
\newcommand{\map}[1]{\mathrm{map}(#1)}
\newcommand{\range}[2]{[#1,#2]}
\newcommand{\substr}[3]{#1[#2..#3]}
\newcommand{\rmq}[3]{\mathrm{RMQ}(#1,#2,#3)}

\newcommand{\depth}[1]{\mathrm{depth}(#1)}
\newcommand{\parent}[1]{\mathrm{parent}(#1)}
\newcommand{\lca}[2]{\mathrm{lca}(#1,#2)}

\newcommand{\preorderselect}[1]{\mathrm{pre\_select}(#1)}

\newcommand{\rleaf}[1]{\mathrm{rleaf}(#1)}

\newcommand{\Talpha}{T_\alpha}
\newcommand{\TalphaL}{\TL{\Talpha}{L}}
\newcommand{\Yalpha}{Y_\alpha}
\newcommand{\Zalpha}{Z_\alpha}

\newcommand{\nodepred}[1]{{#1}_{\mathrm{pred}}}
\newcommand{\nodesucc}[1]{{#1}_{\mathrm{succ}}}
\newcommand{\nodepredsucc}[1]{{#1}_*}

\begin{document}

\title{Succinct data-structure for nearest colored node in a tree}
\author{Dekel Tsur%
\thanks{Department of Computer Science, Ben-Gurion University of the Negev.
Email: \texttt{dekelts@cs.bgu.ac.il}}}
\date{}
\maketitle

\begin{abstract}
We give a succinct data-structure that stores a tree with colors on the nodes.
Given a node $x$ and a color~$\alpha$, the structure finds
the nearest node to $x$ with color~$\alpha$.
This results improves the $O(n\log n)$-bits structure of
Gawrychowski et al.~[CPM 2016].
\end{abstract}

\section{Introduction}
In the \emph{nearest colored node} problem
the goal is to store a tree with colors on the nodes
such that given a node $x$ and a color~$\alpha$, the nearest node to $x$ with
color~$\alpha$ can be found efficiently.
Gawrychowski et al.~\cite{GawrychowskiLMW16} gave a data-structure for this
problem that uses
$O(n\log n)$ bits and answers queries in $O(\log \log n)$ time,
where $n$ is the number of nodes in the tree.
Additionally, they considered a dynamic version of the problem in which
the colors of the nodes can be changed.
For this problem they gave an $O(n\log n)$ bits structure that supports
updates and queries in $O(\log n)$ time.
They also gave a structure with $O(n\log^{2+\epsilon} n)$ space,
optimal $O(\log n/\log \log n)$ query time, and
$O(\log^{1+\epsilon} n)$ update time.

In this paper we give a succinct structure for the static problem.
Our results are given in the following theorem.
\begin{theorem}\label{thm:main}
Let $T$ be a colored tree with $n$ nodes and colors from $\range{1}{\sigma}$,
and let $P_T$ be a string containing the colors of the nodes in preorder.
\begin{enumerate}
\item
For $\sigma = o(\log n/(\log\log n)^2)$, for any $k = o(\log n/\log^2 \sigma)$,
there is a representation of $T$ that uses $n H_k(P_T)+2n+o(n)$ bits
and answers nearest colored node queries in $O(1)$ time,
where $H_k(P_T)$ is the $k$-th order entropy of $P_T$.
\item
For $\sigma=w^{O(1)}$ (where $w$ is the word size),
for any function $f(n) = \omega(1)$,
there is a representation of $T$ that uses $n H_0(P_T)+2n+o(n)$ bits
and answers nearest colored node queries in $O(f(n))$ time.
\item
For $\sigma\leq n$, there is a representation of $T$ that
uses $n H_0(P_T)+2n+o(n H_0(P_T))+o(n)$ bits
and answers nearest colored node queries in $O(\log\frac{\log\sigma}{\log w})$
time.
\end{enumerate}
\end{theorem}

\subsection{Related work}
Several papers studied data-structures for storing
colored trees with support for various
queries~\cite{GearyRR06,BarbayGMR07,FerraginaLMM09,HeMZ14,Tsur_labeled,MuthukrishnanM96,BilleCG14}.
In particular, the problem of finding the nearest ancestor with
color~$\alpha$ was considered
in~\cite{GearyRR06,HeMZ14,Tsur_labeled,MuthukrishnanM96,BilleCG14}.
In order to solve the nearest colored node problem, we combine techniques from
the papers above and from Gawrychowski et al.~\cite{GawrychowskiLMW16}.

Another related problem is to find an approximate nearest node with
color~$\alpha$. This problem has been studied in general 
graphs~\cite{Chechik12,HermelinLWY11,LkackiOPSZ15}
and planar graphs~\cite{AbrahamCKW15,LiMN13,MozesS15}.

\section{Preliminaries}

A node with color~$\alpha$ will be called \emph{$\alpha$-node}.
We also use other $\alpha$- terms with the appropriate meaning,
e.g.\ an $\alpha$-descendant of a node $v$ is a descendant of $v$ (including $v$)
with color~$\alpha$.

\subsection{Sampled RMQ}

A \emph{Range Minimum Query} (RMQ) structure on an array $A$
is a structure that given indices $i$ and $j$, returns an index $k$ such that
$A[k]$ is the minimum element in the subarray $\substr{A}{i}{j}$.
Our data-structure uses sampled RMQ structure, described in the following
lemma.
\begin{lemma}\label{lem:rmq}
Let $A$ be an array of $n$ numbers. For every integer $L$,
there is a data-structure that uses $O((n/L)\log L)$ bits and answers
in constant time RMQ queries $i,j$ in which $i-1$ and $j$ are multiples of $L$.
\end{lemma}
\begin{proof}
Partition $A$ into blocks of size $L$, and let $A'$ be an array of size
$\lceil n/L\rceil$ in which $A[i]$ is the minimum element in the $i$-th block
of $A$.
The data-structure stores the RMQ structure of Fischer~\cite{Fischer10} on $A'$
(note that this structure does not access $A'$ in order
to answer queries),
and an array $B$ in which $B[i]$ is the index of the
minimum element in the $i$-th block of $A$, relative to the start of the block.
The space for the RMQ structure is $2n/L+o(n/L)$ bits, and the space for $B$
is $O((n/L)\log L)$ bits. Thus, the total space is $O((n/L)\log L)$ bits.
\end{proof}

\subsection{Tree decomposition}\label{sec:tree-decomposition}
We use the following tree
decomposition~\cite{AlstrupHLT97,AlstrupSS97,FarzanM08}.
\begin{lemma}\label{lem:tree-decomposition}
For a tree $T$ with $n$ nodes and an integer $L$, there is a collection
$\D{T}{L}$ of subtrees of $T$ with
the following properties.
\begin{enumerate}
\item
Every edge of $T$ appears in exactly one tree of $\D{T}{L}$.
\item
The size of every tree in $\D{T}{L}$ is at most $L$ and at least~$2$.
\item
The number of trees in $\D{T}{L}$ is $O(n/L)$.
\item\label{enum:boundary}
For every $T'\in \D{T}{L}$, at most two nodes of $T'$ can appear
in other trees of $\D{T}{L}$.
These nodes are called the \emph{boundary nodes} of $T'$.
\item\label{enum:boundary2}
A boundary node of a tree $T'\in \D{T}{L}$ can be either a root of $T'$
or a leaf of $T'$.
In the latter case the node will be called the \emph{boundary leaf} of $T'$.
\end{enumerate}
\end{lemma}
For a tree $T$ and an integer $L$ we define a tree $\TL{T}{L}$ as follows.
Construct a tree decomposition $\D{T}{L}$ according to
Lemma~\ref{lem:tree-decomposition}.
If the root $r$ of $T$ appears in several trees of $\D{T}{L}$,
add to $\D{T}{L}$ a tree that consists of $r$.
The tree $\TL{T}{L}$ has a node $v_S$ for every tree $S \in \D{T}{L}$.
For two trees $S_1,S_2\in\D{T}{L}$, $v_{S_1}$ is the parent of $v_{S_2}$
in $\TL{T}{L}$ if and only if the root of $S_2$ is equal to the boundary leaf of
$S_1$.

\section{Structure for small alphabet}\label{sec:small-alphabet}

In this section we prove part~1 of Theorem~\ref{thm:main}.
Our structure is similar to the labeled tree structure of
He et al.~\cite{HeMZ14}.
As in~\cite{HeMZ14}, the data-structure stores $P_T$
in the compressed structure of Ferragina and Venturini~\cite{FerraginaV07},
the balanced parenthesis string of $T$ (uncompressed),
and additional information that takes $o(n)$ space.
The space for storing $P_T$ is $nH_k(P_T)+o(n)$ bits,
and the space for storing the balanced parenthesis string is $2n$ bits.
Using $o(n)$ space, standard tree queries such as computing the depth of a node
or computing the lowerst common ancestor of two nodes
can be answered in $O(1)$ time (see~\cite{FarzanM08} for details).

Using the tree decomposition of Lemma~\ref{lem:tree-decomposition},
the tree $T$ is partitioned into \emph{mini-trees} of
size at most $L'=\log^{\Theta(1)} n$, and every mini-tree
is decomposed into \emph{micro-trees} of size at most
$L = \Theta(\log_\sigma n)$.
From Property~\ref{enum:boundary} of the tree decomposition we have that
if $y$ is the nearest $\alpha$-node to a node $x$
and $y$ is not in the same mini-tree of $x$,
then the path from $x$ to $y$ passes through a boundary node
of the mini-tree of $x$.
Similar property holds for micro-trees.
Based on the observation above,
the data-structure stores the following additional information.
\begin{inparaenum}[(1)]
\item
A lookup table that contains for every colored tree $S$ of size at most $L$,
every node $x$ in $S$, and every color~$\alpha$,
the $\alpha$-node in $S$ that is nearest to $x$ (if such nodes exist).
\item
For every mini-tree $S$ and every color~$\alpha$,
the $\alpha$-nodes in $S$ that are nearest
to the root of $S$ and to the boundary leaf of $S$.
\item
For every micro-tree $S$ and every color~$\alpha$,
the $\alpha$-nodes in the mini-tree of $S$ that are nearest
to the root of $S$ and to the boundary leaf of $S$.
\end{inparaenum}
The space for the lookup table is $O(2^{2L} \sigma^{L} L\sigma\log L) =
o(n)$,
the space for the mini-tree information is $O((n/L')\sigma\log n) = o(n)$, and
the space for the micro-tree information is
$O((n/L)\sigma\log L')
= O(\sigma\log\sigma\cdot n\log\log n/\log n)
= o(n)$.

Given a query $x,\alpha$, we obtain up to five candidates for the nearest
$\alpha$-node to $x$:
the $\alpha$-node in the micro-tree of $x$ that is nearest to $x$,
and the four $\alpha$-nodes stored for the micro-tree and mini-tree of $x$.
Note that in order to use the lookup-table, we need to generate
the balanced parenthesis string of the the micro-tree of $x$,
and a sequence containing the colors of the nodes in this tree in preorder.
This can be done in $O(1)$ time due to property~\ref{enum:boundary2} of
Lemma~\ref{lem:tree-decomposition}
(this property implies that for a tree $T'\in  \D{T}{L}$ there are two
intervals $I_1$ and $I_2$ such that a node $x\in T$ is
a non-root node of $T'$ if and only if the preorder rank of $x$ is in
$I_1 \cup I_2$).
The distance between $x$ and every candidate $y$ can be computed in
$O(1)$ time (the distance is $\depth{x}+\depth{y}-2\cdot\depth{\lca{x}{y}}$).
Therefore, the query is answered in $O(1)$ time.

\section{Structure for large alphabet}\label{sec:large-alphabet}

Our data-structure for large alphabet stores
the rank-select structure of Belazzougui and Navarro~\cite{BelazzouguiN12}
on $P_T$,
a succinct tree structure for the tree $T$ without the colors,
and additional information that will be described below.
The space of the rank-select structure is
$n H_0(P_T)+o(n)$ for $\sigma=w^{O(1)}$,
and $n H_0(P_T)+2n+o(n H_0(P_T))+o(n)$ for general $\sigma$.
The space for storing $T$ is $2n+o(n)$.

Our structure is similar to the structure of
Gawrychowski et al.~\cite{GawrychowskiLMW16}.
We next give a short description of the structure of~\cite{GawrychowskiLMW16}.
For a color~$\alpha$,
let $\Zalpha$ be the set of all $\alpha$-nodes and their ancestors,
and let $\Yalpha$ be the set of all nodes $x \in \Zalpha$ such that
either $x$ has color~$\alpha$, or $x$ has at least two children in $\Zalpha$.
We define a tree $\Talpha$ whose nodes are $\Yalpha$, and $x$ is the parent of
$y$ in $\Talpha$ if and only if $x$ is the lowest proper ancestor of $y$
that is in $\Yalpha$.

Let $x,\alpha$ be a query.
If $x$ is an $\alpha$-node, the answer to the query is trivial, so assume for
the rest of the section that $x$ is not an $\alpha$-node.
We define nodes in the tree that will be used for answering the query:
$z$ is the lowest ancestor of $x$ which has an $\alpha$-descendant
that is not a descendant of $x$.
Moreover, $y$ (resp., $y_2$) is the lowest descendant (resp., ancestor) of $z$ which
is in $\Yalpha$.
The nearest $\alpha$-node to $x$ is either
\begin{inparaenum}[(1)]
\item
the nearest $\alpha$-descendant of $x$,
\item
the nearest $\alpha$-node to $y$, or
\item
the nearest $\alpha$-node to $y_2$.
\end{inparaenum}
Based on this observation, the structure of
Gawrychowski et al.~\cite{GawrychowskiLMW16}
finds these three candidate nodes, and returns the one that is closest to $x$.
Our structure is based on a slightly different observation:
The nearest $\alpha$-node to $x$ is either
\begin{inparaenum}[(1)]
\item
the nearest $\alpha$-descendant of $x$,
\item
the nearest $\alpha$-descendant of $y$, or
\item
the nearest $\alpha$-non-descendant of $y$.
\end{inparaenum}

We next describe the approach we use for finding the node $y$,
which is different than the one used in~\cite{GawrychowskiLMW16}.
For a node $v$, let $\nodepred{v}$ (resp., $\nodesucc{v}$) be the last
(resp., first) $\alpha$-node in preorder that appears before
(resp., after) $v$ in the preorder.
The following lemma shows how to efficiently find $z$.
\begin{lemma}\label{lem:z}
Let $\nodepredsucc{x}$ be the node from $\{\nodepred{x},\nodesucc{x}\}$
that maximizes the depth of $\lca{\nodepredsucc{x}}{x}$
(if $\nodepred{x}$ does not exist then $\nodepredsucc{x}=\nodesucc{x}$
and vice versa).
Then, $z = \lca{\nodepredsucc{x}}{x}$.
\end{lemma}
\begin{proof}
We will prove the lemma for the case in which $\nodepred{x}$ and $\nodesucc{x}$
exist and $\nodepred{x}$ is not an ancestor of $\nodesucc{x}$.
The proofs for the other cases are similar and thus omitted.

Let $x' = \lca{\nodepred{x}}{\nodesucc{x}}$. 
Let $\nodepred{P}$ and $\nodesucc{P}$ be the paths from $x'$ to
$\nodepred{x}$ and $\nodesucc{x}$, respectively.
The nodes that appear between
$\nodepred{x}$ and $\nodesucc{x}$ in preorder can be categorized into five
sets:
\begin{enumerate}
\item All proper descendants of $\nodepred{x}$.
\item All proper descendants of nodes on the path $\nodepred{P}$
that are ``to the right'' of this path.
Formally, this set contains every node $v$ which is a descendant of a node $w$
on $\nodepred{P}$ (excluding the endpoints)
such that the child of $w$ which is on the path from $w$ to
$v$ is to the right of the child of $w$ which is on the path $\nodepred{P}$.
\item All proper descendants of $x'$ that are
``between the paths'' $\nodepred{P}$ and $\nodesucc{P}$.
\item All proper descendants of nodes on the path $\nodesucc{P}$
(excluding the endpoints) that are ``to the left'' of this path.
\item All nodes on the path $\nodesucc{P}$ (excluding the endpoints).
\end{enumerate}
Since $\nodepred{x}$ and $\nodesucc{x}$ are consecutive $\alpha$-nodes in the
preorder of the nodes, the nodes in the sets above are not $\alpha$-nodes.
Moreover, the nodes in sets 1--4 that are not descendants of $x$
are not in $\Zalpha$.
Suppose $x$ is in set~2.
In this case, $\lca{\nodepred{x}}{x}$ is on $\nodepred{P}$ while
$\lca{\nodesucc{x}}{x} = x'$.
Therefore, $\nodepredsucc{x} = \nodepred{x}$.
Since $\lca{\nodepred{x}}{x} \in \Zalpha$ (as $\lca{\nodepred{x}}{x}$ is an
ancestor of $\nodepred{x}$) and every ancestor $x'$ of $x$ that is below
$\lca{\nodepred{x}}{x}$ does not have an $\alpha$-descendant that is not a
descendant of $x$,
it follows that $\lca{\nodepred{x}}{x} = z$.
Thus, the lemma follows in this case.
The proofs for the other cases are similar.
\end{proof}
Finding $\nodepred{x}$ and $\nodesucc{x}$ can be done using rank
and select queries on $P_T$.
The next lemma shows how to find $y$.
\begin{lemma}\label{lem:w}
Let $v$ be a node in $\Zalpha$, and let $w$ be the highest descendant of $v$
that is in $\Yalpha$.
Then, $w=\lca{\nodesucc{v}}{\nodepred{\rleaf{v}}}$,
where $\rleaf{v}$ is the rightmost descendant leaf of $v$.
\end{lemma}
\begin{proof}
Suppose first that $w$ does not have color~$\alpha$.
Since $w \in \Yalpha$, $w$ has at least two children that are in $\Zalpha$.
Let $w',w''$ be the first and last children of $w$ that are in $\Zalpha$,
respectively.
Every $\alpha$-descendant of $v$ is also a descendant of $w$.
It follows that $\nodesucc{v}$ is a descendant of $w'$ and
$\nodepred{\rleaf{v}}$ is a descendant of $w''$.
Thus, $\lca{\nodesucc{v}}{\nodepred{\rleaf{v}}} = \lca{w'}{w''} = w$.

If $w$ has color~$\alpha$ then $\nodesucc{v} = w$ and $\nodepred{\rleaf{v}}$ is
a descendant of $w$.
Therefore, $\lca{\nodesucc{v}}{\nodepred{\rleaf{v}}} = w$.
\end{proof}

We now show how to find the nearest $\alpha$-non-descendant of $y$.
We use the approach of~\cite{Tsur_labeled}.
For the case $\sigma=w^{O(1)}$ let $L = f(n)$ where $f$ is a function that
satisfies $f(n)=\omega(1)$ and $f(n)=O(\log n)$,
and for larger $\sigma$ let $L = \sqrt{\log\frac{\log\sigma}{\log w}}$.
We say that a color~$\alpha$ is \emph{frequent} if the number of $\alpha$-nodes
is at least $L$.
Given a query $x,\alpha$, finding whether $\alpha$ is frequent can be done
by performing a rank query on $P_T$.
If $\alpha$ is non-frequent, the query can be answered
by enumerating all $\alpha$-nodes
(by computing $\preorderselect{\select{\alpha}{P_T}{k}}$ for all $k$,
where $\select{\alpha}{P_T}{k}$ is the $k$-th occurrence of $\alpha$ in $P_T$
and $\preorderselect{i}$ is the $i$-th node of $T$ in preorder)
and computing the distance between $x$ and each enumerated node.
The time complexity is $O(L\cdot\tselect)$,
where $\tselect$ is the time of a select query on $P_T$.
Since $\tselect = O(1)$ for small alphabet
and $\tselect = o(\sqrt{\log\frac{\log\sigma}{\log w}})$ for large alphabet,
it follows that the time for answering a query is any $\omega(1)$ for small
alphabet and $O(\log\frac{\log\sigma}{\log w})$ for large alphabet.
For the rest of the section, we describe how to handle queries in which
the color is frequent.

We apply the tree decomposition of Section~\ref{sec:tree-decomposition} on
$\Talpha$ with parameter $L$ and obtain the tree $\TalphaL$.
For a node $v_S$ in $\TalphaL$
(recall that $S$ is a subtree of $\Talpha$) let
$\V{v_S}$ be the set of the nodes of $S$ excluding the root, and
$\Va{v_S}$ be the set of $\alpha$-nodes in $\V{v_S}$.
For a node $u\in \Yalpha$, we denote by $\map{u}$ the
node of $\TalphaL$ for which $u\in \V{\map{u}}$.
Due to the properties of the tree decomposition we have that for two nodes
$u,v \in \Yalpha$, $\lca{u}{v}$ is a node in the tree $S$ in the 
decomposition for which $v_S = \lca{\map{u}}{\map{v}}$
(if $\lca{u}{v}$ is not the root of $S$ then $\map{\lca{u}{v}} = v_S$
and otherwise $\map{\lca{u}{v}} = \parent{v_S}$).


We assign weights to each node $v_S$ of $\TalphaL$ as follows.
\begin{itemize}
\item $\wdist(v_S)$ is the distance between the boundary nodes of $S$.
\item $\wdistroot(v_S)$ (resp., $\wdistspecial(v_S)$) is the
shortest distance between the root (resp., boundary leaf) of $S$ and
a node in $\Va{v_S}$. If $\Va{v_S} = \emptyset$ then
$\wdistroot(v_S) = \wdistspecial(v_S) = \infty$.
\end{itemize}
Let $v$ and $v'$ be two nodes of $\TalphaL$,
and let $P$ be the path from $v$ to $v'$.
The \emph{weighted distance from $v$ to $v'$} is the sum of the following
values.
\begin{enumerate}
\item $\wdist(u)$ for every node $u \neq v,v'$ which is on $P$
and the parent of $u$ is also on $P$.
\item $\wdistroot(v')$ if $v'$ is not an ancestor of $v$.
\item $\wdistspecial(v')$ if $v'$ is an ancestor of $v$.
\end{enumerate}
The descendant (resp., non-descendant) of $v$ with minimum weighted distance of $v$
(with ties broken arbitrarily) will be denoted $\wnearestd{v}$ (resp., $\wnearestnd{v}$).

Our approach for finding the nearest $\alpha$-non-descendant of $y$ is based on
the following observation.
\begin{observation}
Let $v_S \neq v_{S'}$ be two nodes of $\TalphaL$,
and $u$ be a node in $S$. 
The shortest distance between $u$ and a node in $\Va{v_{S'}}$ is equal
to the weighted distance from $v_S$ to $v_{S'}$ plus
the distance between $u$ and the boundary leaf of $S$ if $v_{S'}$ is a
descendant of $v_S$,
and the distance between $u$ and the root of $S$ otherwise.
\end{observation}
\begin{corollary}\label{cor:yp-v}
Let $v_S$ be a node of $\TalphaL$.
Let $u$ be a node in $S$ 
and $u'$ be the nearest $\alpha$-non-descendant of $u$.
If $u$ is not the root of $S$ then
$u' \in \Va{v_S} \cup \Va{\wnearestnd{v_S}} \cup \Va{\wnearestd{v_S}}$
and otherwise
$u' \in \Va{\parent{v_S}} \cup \Va{\wnearestnd{\parent{v_S}}}$.
\end{corollary}
Based on Corollary~\ref{cor:yp-v},
the algorithm for finding the nearest $\alpha$-non-descendant of $y$ is as
follows.

\begin{enumerate}
\item Find $z$ using Lemma~\ref{lem:z}.
\item Compute $y = \lca{\nodesucc{z}}{\nodepred{\rleaf{z}}}$.
\item Compute $y' = \lca{\map{\nodesucc{z}}}{\map{\nodepred{\rleaf{z}}}}$
and $y'' = \parent{y'}$.\label{enu:mapy}
\item Enumerate all nodes in\label{enu:enumerate}
$\Va{y'} \cup \Va{\wnearestnd{y'}} \cup \Va{\wnearestd{y'}} \cup
 \Va{y''} \cup \Va{\wnearestnd{y''}}$ that are not descendants of $y$,
compute their distances to $y$, and return the node that is nearest to $y$.
\end{enumerate}

In order to perform steps~\ref{enu:mapy} and~\ref{enu:enumerate} efficiently,
we store $T'$ using the weighted tree structure of~\cite{Tsur_labeled}.
This structure supports computing $\map{u}$ for an $\alpha$-node~$u$
in $O(1)$ time,
and additionally, it supports computing the preorder ranges of the nodes
in $\Va{v_S}$ for a node $v_S$ of $\TalphaL$ in $O(1)$ time.

The data-structures for computing $\wnearestd{v_S}$ and $\wnearestnd{v_S}$
for some node $v_S$ of $\TalphaL$ is as follows.
We apply the tree decomposition of
Section~\ref{sec:tree-decomposition}.
Every tree $\TalphaL$ is partitioned into \emph{mini-trees}
of size at most $L_1=\log^{\Theta(1)} n$, and every mini-tree
is decomposed into \emph{micro-trees} of size at most $L_2 = \Theta(\log n)$.
The trees $\TalphaL$ are merged into a single tree $T'$ by connecting their
roots to a new node.
We now store the following.
\begin{itemize}
\item
A lookup table that contains for every tree $S$ of size at most $L$
with weights $\wdist,\wdistroot,\wdistspecial$
on its nodes and every node $u$ in $S$, the node $u_2$ in $S$ whose weighted
distance from $u$ is minimum.
\item
The balanced parenthesis sequence of $T'$.
\item
For $i=\idxmin,\ldots,\idxmax$,
a strings $W_i$ that contains the $w_i$ weights of the nodes of $T'$ according
to preorder.
\item
For every mini-tree $S$ and every color~$\alpha$,
the nodes in $\TalphaL$ with minimum weighted distances
to the root of $S$ and to the boundary leaf of $S$.
\item
For every micro-tree $S$ and every color~$\alpha$,
the nodes in the mini-tree of $S$ with minimum weighted distances
to the root of $S$ and to the boundary leaf of $S$.
\end{itemize}
Using this information, $\wnearestd{v_S}$ and $\wnearestnd{v_S}$ can be found in $O(1)$ time.
Therefore, the nearest $\alpha$-non-descendant of $y$ 
can be found in time $\omega(1)$ for small alphabet and
$O(\log\frac{\log\sigma}{\log w})$ for large alphabet.

Finally, we describe how to find the nearest $\alpha$-descendant of a node $v$.
For every frequent color~$\alpha$, let $A_\alpha$ be an array containing
the depths of the $\alpha$-nodes in preorder.
We build a sampled RMQ structure (Lemma~\ref{lem:rmq})
on $A_\alpha$ with sampling parameter $L$.
To find the nearest $\alpha$-descendant of a node $v$,
find the range $\range{i}{j}$ of preorder ranks of
the $\alpha$-descendants of $v$ using rank queries on $P_T$ and tree queries
on $T$.
Let $i'$ be the minimum integer such that $i' \geq i$ and $i'-1$ is a multiple
of $L$, and
let $j'$ be the maximum integer such that $j' \leq j$ and $j'$ is a multiple
of $L$.
Assuming $i'<j'$ (the case $i'>j'$ is simpler and we omit the details),
enumerate the $\alpha$-nodes with preorder ranks in
$\range{i}{i'-1} \cup \range{j'+1}{j} \cup \{\rmq{A_\alpha}{i'}{j'}\}$,
compute the distances between these node and $v$,
and return the node with smallest distance.

\bibliographystyle{plain}
\bibliography{string-index,nearest,dekel}
\end{document}